\documentclass[10pt,conference]{IEEEtran}
\usepackage{amssymb,color,amsmath}
\usepackage{cite,manfnt}
\usepackage{paralist}
\usepackage[colorlinks=true,urlcolor=rltblue]{hyperref}
\usepackage{color}
\definecolor{rltblue}{rgb}{0,0,0.75}
\usepackage{graphicx}

\newtheorem{theorem}{Theorem}
\newtheorem{lemma}[theorem]{Lemma}
\newtheorem{remark}[theorem]{Remark}
\newtheorem{defn}[theorem]{Definition}

\newcommand{\E}{\mathbf{E}}

\newcommand{\nix}[1]{}

\begin{document}
\title{Bounds on the Network Coding Capacity for\\  Wireless Random Networks}
\author{Salah A. Aly, Vishal Kapoor, Jie Meng, Andreas Klappenecker\\
 \emph{Department of Computer Science,}\\
 \emph{Texas A\&M University}\\
 \emph{College Station, TX - 77843}\\
\{salah, vishal, jmeng, klappi\}@cs.tamu.edu}

\maketitle

\begin{abstract}
Recently, it has been shown that the max flow capacity can be achieved in a
multicast network using network coding. In this paper, we propose and analyze a more
realistic model for wireless random networks. We prove that the capacity of network
coding for this model is concentrated around the expected value of its minimum cut.
Furthermore, we establish  upper and lower bounds for  wireless nodes using Chernoff
bound. Our experiments show that our theoretical predictions are well matched by
simulation results.
\end{abstract}

\section{Introduction}\label{intro}

Traditionally, the information flow in networks is modeled as a
multi-commodity flow problem by treating the underlying network as a
flow network. Suppose that one source node in a graph has to transfer
some information to one destination node (i.e., a unicast
situation). By Menger's theorem\cite{Diestal05}, the maximum
information that can flow is upper bounded by the value of the minimum
cut between the source and the destination; this well-known result
from classical graph theory is also known as the \emph{Max-flow
Min-Cut theorem}.  One can use max-flow min-cut algorithms to compute
the maximum throughput for instance for unicast, multicast, and
multi-source multicast communications.

Recently, Ahlswede, Cai, Li, and Yeung proposed in the seminal
paper\cite{Ahlswede00} a new paradigm, called \emph{network coding}.
Their key observation was that traditional store-and-forward networks
cannot always achieve the max-flow value, whereas one can achieve this
value using network coding. The idea is based on the simple fact that
information can be replicated, mixed together and then transmitted
over links to save bandwidth. If this is properly done, then the
information can be reliably decoded at the receiver nodes, see
e.g.~\cite{Ahlswede00,Yeung02}.  The basic idea of network coding is
that the intermediate network nodes can now process, encode, and
transmit information.

Since its inception by Ahlswede et al., there has been an upsurge of
interest in network coding, see for
example~\cite{Fragouli06,Ho03routing,Ho06,Ho03randomized,
Jaggi05,Koetter03,Li03} and the references therein.  Arguably, most
network coding publications model the underlying network as a directed
acyclic graph and are typically concerned with solving single source
multicast or multi-source multicast using deterministic or randomized
encoding and decoding schemes.

In this paper, we discuss a new model for wireless random networks. In this model,
nodes are placed at random locations. Two nodes $u$ and $v$ are connected with
probability 1 if the distance between them is less than or equal to $r$; the nodes
are connected with probability $p<1$ if the distance between them is less than or
equal to $R$ but greater than $r$; otherwise $u$ and $v$ are not connected. Thus,
the model is a refinement of geometric random graphs that incorporates the potential
loss of connectivity towards the end of the transmission range, where interference
is more dominant. The main contributions of this paper are:
\begin{itemize}
\item We introduce the quasi random geometric graph model, a
model of wireless network topologies that simulates the connectivity
in mobile ad-hoc networks more realistically than the random graph
model, but is still easy to analyze.
\item We derive high-probability bounds for the network coding
capacity of quasi random geometric graphs.
\item We provide simulations results that support our bounds on the
network coding capacity.
\end{itemize}
The rest of this paper is organized as follows. In Section \ref{bmd}, we give
an overview of network coding and the previous work in capacity of network
coding. In Section \ref{model}, we present our new model. We provide our main
results in Sections~\ref{bounds} and~\ref{simulation}.

\section{Background and Model Description}\label{bmd}
In this section, we give a short summary of network coding, focusing
on the calculation of the capacity of a min cut in a weighted random
graph.  For a more in depth discussion of basic concepts and methods
of network coding, we refer the reader to the survey
paper~\cite{Fragouli06}.

\subsection{Network Coding Fundamentals}
To illustrate the power of network coding, we provide a simple
example, which is often referred to as the \emph{Wheatstone bridge},
due to its electrical circuits origin. It demonstrates that multicast
routing can achieve the maximum possible throughput in a communication
network using a coding scheme consisting of linear operations in
finite field, whereas traditional store-and-forward routing cannot
achieve the same throughput.

Consider the example shown in Fig.\ref{fig:wheatstone_bridge}(a),
where the nodes $X$ and $Y$ respectively want to send two bits $b_{1}$
and $b_{2}$ to each other. One way of doing this is to let the bit
$b_{1}$ travel on the path $X\rightarrow A\rightarrow B\rightarrow Y$
at one point of time and to let $b_{2}$ travel on the path
$Y\rightarrow A\rightarrow B\rightarrow X$ on the other. However, if
the network wants to transmit the bits simultaneously, then there is
no way to do so, as there are no disjoint paths between $X$ and~$Y$.

However, using network coding as
shown in Fig. \ref{fig:wheatstone_bridge}(b), one can save bandwidth.
In this case, both $X$
and $Y$ transmit the bits $b_{1}$ and $b_{2}$ (as shown in the figure)
and then $A$ XORs (\emph{encodes}) them together and the resulting bit
$b_{1}\oplus b_{2}$ travels over the paths $A\rightarrow B\rightarrow
Y$ and $A\rightarrow B\rightarrow X$. Since node $X$ already has of
$b_{1}$, it can recover (\emph{decode}) $b_{2}$ by the operation
$b_{1}\oplus(b_{1}\oplus b_{2})$. Similarly $Y$ can also decode
$b_{1}$.

This example illustrates that the capacity of the minimum cut (equal
to 1 in this example) can be easily achieved by network coding,
whereas two rounds are needed to achieve the multicast in the uncoded
(traditional) routing case, assuming unit capacity edges.
Because of such benefits, network coding
can be used in wireless ad-hoc networks or sensor networks to help
conserve energy and to increase the overall throughput.

\begin{figure}[htbp]
\begin{center}
\includegraphics[scale=0.3,width=8cm,height=4cm]{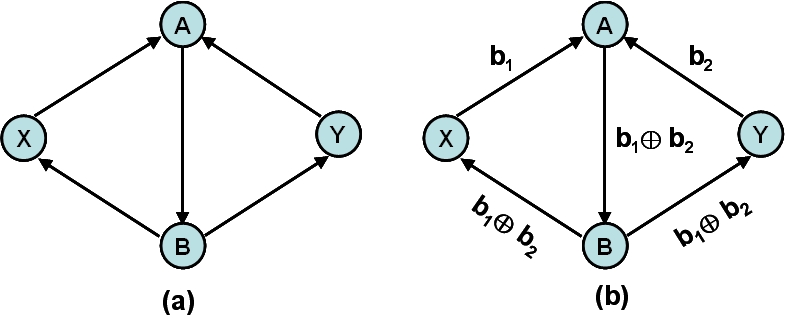}
    \caption{An example of network coding on a Wheatstone Bridge}
    \label{fig:wheatstone_bridge}
\end{center}
\end{figure}

\subsection{Network Coding in Ad-hoc Wireless Networks}

In \cite{Ramamoorthy05}, Ramamoorthy et al. modeled the capacities of the
connected edges in a wireless network as a Weighted Random Geometric Graph
($\mathcal{G}^{WRGG}$) and considered the \emph{single source multicast
problem}.
\begin{defn}[Single Source Multicast Problem]
Let $G=(V,E)$ be a graph with vertex set $V$ and edge set $E$
representing a network. Let $S \subseteq V$ be a set of sources
(origins) and $T \subseteq V$ be a set of terminals (destinations). The
multicast problem is to distribute the messages from the senders $s\in
S$ to all terminal nodes $t\in T$, allowing routing along the edges of
$G$. In network coding, the vertices are allowed to encode the
incoming bits (or packets) and send encoded versions along the
outgoing edges.  A \emph{single source multicast problem} is the
special case where one has a single sender, that is, $|S|=1$.
\end{defn}
\smallskip

Ramamoorthy et al. extended the results proved by Karger et al. in
\cite{Karger99} and used them to derive bounds for coding capacity for
a single source multicast problem in a network comprised of a single
source $s$, an intermediate network consisting of $n$ relay nodes, and
$l$ terminal nodes, having independent and identically distributed
link capacities $\sim X$ between any two nodes.  They showed that the
network coding capacity is concentrated around the value $n\E[X]$ in
such a network.

In this paper, we extend their work to a more general and more
realistic model that we call the \emph{Quasi Random Geometric Graph}
model ($\mathcal{G}^{QRGG}$). We derive high-probability bounds for
the network coding capacity of such graphs.

\section{Modeling Random Wireless Networks}\label{model}
In this section, we present our new model and study the capacity of a minimum
cut in a random wireless network.

Let $r$ be  a real number in the range $0 \leq r \leq 1$.  Recall that a Random
Geometric Graph is a graph $\mathcal{G}^{RGG}=(V,E)$ with $n$ nodes selected
independently and uniformly at random from the unit square $[0,1]^2$ in which
any two nodes $u$ and $v$  in $V$ are connected by an edge $(u,v)$ in $E$ if
and only if the Euclidean distance $d(u,v) \leq r$. Such a graph is  rough
approximation of wireless networks.

Random geometric graphs have been popular in wireless mobile ad-hoc networks
literature, since it is a theoretical model of the network topology that is
easy to analyze. However, it does not realistically model the area of
transmission, which is, in general, not a disk of radius $r$. Recently, a more
realistic model for connectivity was proposed by Kuhn, Wattenhofer, and
Zollinger~\cite{Kuhn03}. In their model, two nodes $u$ and $v$ may or may not
be connected when their Euclidean distance $d(u,v)$ is within the range
$r<d(u,v)\leq r'$, see Fig. \ref{fig:qrgg}. We use random instances of such
quasi-disk graphs to model the dynamically changing network topology in
wireless random ad-hoc networks.

\begin{figure}[htbp]
\begin{center}\includegraphics[scale=0.3,width=4cm,height=4cm]{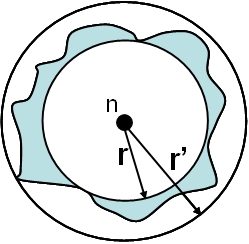}
    \caption{The transmission range and Quasi
    Disk Graph Representation for a node}
    \label{fig:qrgg}
\end{center}
\end{figure}

\begin{defn}[Quasi Random Geometric Graph ($\mathcal{G}^{QRGG}$)]
Let $r$ and $r'$ be two real numbers in the range $0\le r<r'\le
1$. Let $V$ be a set of $n$ nodes that are selected independently and
uniformly at random from the unit square $[0,1]^2$. If $u$ and $v$ are
two nodes in $V$, then
\begin{enumerate}
    \item $(u, v) \in E$ if $d(u,v) \leq r$;
    \item $(u, v) \notin E$ if $d(u,v) > r'$;
    \item $(u, v) \in E$ with probability $p$ if $r<d(u,v)\leq r'$.
\end{enumerate}
We call $\mathcal{G}^{QRGG}=(V,E)$ a quasi random geometric graph.
\end{defn}
\smallskip

The difference between quasi random geometric graphs and random
geometric graphs is that nodes at distance $d$ within the range
$r<d\leq r'$ may or may not be connected; this models the connectivity
in a more realistic way.

\begin{remark}\label{rem:var_prob}
Instead of having a fixed probability $p$ for the connectivity of
nodes within distance $d$ in the range $r<d\le r'$, one can use a
function $p(d)$ that associates a probability that depends on the
distance to model the attenuation of the signal. Such a change is of
course straightforward. We give one example in
Section~\ref{simulation}.
\end{remark}

In this paper, we consider the problem of
single-source multicasts in such quasi random geometric graphs. Our
main concern is to provide a lower bound for the capacity of network
coding in this situation.  Before defining the capacity, we need to
further detail our model of connectivity.

\begin{defn}[Connectivity Graph]
Let $s$ be a source node, $T$ a set of terminal nodes, and $R$ a set
of relay nodes. We define a connectivity graph $G=(V,E)$ as a graph
with vertex set $V=\{s\}\cup R\cup T$ such that $G\in
\mathcal{G}^{QRGG}$; in particular, the vertices are located in a
unit square.
We assume further that the source node only
sends messages and terminal nodes only receive messages; in
particular, the source and terminal nodes do not relay any messages.
Furthermore, we assume that the source and the terminal nodes do not
communicate directly; thus, any message is routed through at least one
relay node.
\end{defn}

We assume that the edges in the connectivity graph represent links
with unit capacity. Put differently, we assume that the capacity
$C_{ij}$ for $i, j$ in $V$ is given by
$$
C_{ij}=\left\{
\begin{array}{ll}
1 & \hbox{if $(i,j)\in E$,} \\
0 & \hbox{otherwise.}
\end{array}\right.
$$
We note that $C_{ij}=C_{ji}$, since the graph is undirected.

\begin{defn}[A Cut and its Capacity]\label{defn:connect}
Let $G=(V,E)$ be a connectivity graph with source node $s$, a set $T$
of terminal nodes, and a set $R$ of relay nodes such that
$V=\{s\}\cup R\cup T$. Let $t$ be a terminal node in $T$.
An $s$-$t$-cut  of size $k$ in the connectivity graph $G$ is a partition of
the set of relay nodes $R$ into two sets $V_k$ and $\overline{V}_k$ such that
\begin{enumerate}
\item[(i)] $|V_{k}|=k$ and $|\overline{V}_k|=n-k$;
\item[(ii)] $R=V_{k}\cup \overline{V_k}$ and $V_k\cap
\overline{V_k}=\emptyset$.
\end{enumerate}
The edges crossing the cut are given by
\begin{enumerate}
\item $E\cap \{ (s,i) |  i\in \overline{V}_k\}$;
\item $E\cap \{ (j,t) | j \in V_k\}$;
\item $E\cap \{ (j,i) | j \in V_k \text{ and } i\in \overline{V_k}\}$.
\end{enumerate}
In other words, the source node $s$ and the relay nodes $V_k$ are on
one side of the cut, whereas the relay nodes $\overline{V_k}$ and the
terminal node $t$ on the other side of the cut. The total capacity of
an $s$-$t$-cut of size $k$ is given by
\begin{eqnarray}\label{eq:cut}
    C_{k}= \sum_{i \in \overline{V_k}}
    C_{si}
      +\sum_{j \in V_k\vphantom{|}}\sum_{i\in \overline{V_k}} C_{ji}
      +\sum_{j \in V_k}C_{jt}.
\end{eqnarray}
\end{defn}

\section{Bounds and Results}\label{bounds}
In this section, we bound the network coding
capacity of a connectivity graph, where the connections of the relay
nodes form an instance of a quasi random geometric graph.

Let $G=(V,E)$ be a connectivity graph such that the vertex set $V$
consists of a source node $s$, a set of terminal nodes $T$, and a set
of relay nodes $R$, that is, $V=\{s\}\cup T\cup R$. Recall that two
nodes $u$ and $v$ in $G$ are connected by an edge with
probability 1 if $d(u,v)\le r$, with probability $p$ if $r<d(u,v)\leq
r'$, and with probability $0$ otherwise. Therefore, the probability
$p'$ that two nodes $u$ and $v$ are connected can be bounded by
\begin{equation}\label{eq:conn_prob}
\frac{1}{4}\left(\pi r^2+\pi(r'^2-r^2)p\right) \le p' \le \pi
r^2+\pi(r'^2-r^2)p.
\end{equation}
The motivation for the lower bound stems from
the fact that one of the nodes might be located in one of the corners
of the unit square. The upper bound is a straightforward consequence
of our connectivity rules.

These elementary observations allow us to bound the expected value
of the cut $C_k$. By equation~(\ref{eq:cut}), we have
$$
\begin{array}{lcl}
\E[C_k] &=& \displaystyle\sum_{i\in \overline{V_k}} \E[C_{si}]
+\sum_{j\in V_k}\sum_{i\in\overline{V_k}} \E[C_{ji}]
+ \sum_{j\in V_k} \E[C_{jt}]\\
&=& p'(n+k(n-k)).
\end{array}
$$ In particular, $\E[C_k]=\E[C_{n-k}]$ holds for all $k$ in the range $0\le
k\le n$.  Furthermore, we have
$$\E[C_0]=\E[C_n]\le \E[C_1]=\E[C_{n-1}]\le \cdots \le \E[C_{\lceil
n/2\rceil}].$$

Our goal is to prove that the capacity $C_k$ of an $s$-$t$-cut is
concentrated around its expected value.  A technical difficulty arises
because the edges between relay nodes in the graph $G$ are in general
not mutually independent.  Indeed, if two relay nodes $u$ and $v$ are
connected, and $u$ is connected to yet another relay node $w$, then
there is a good chance that $v$ is connected to~$w$. Put differently,
we have
$$\Pr[(v,w)\in E| (u,v)\in E, (u,w)\in E] > \Pr[(v,w)\in E],$$ whence
the three events $(u,v)\in E$, $(u,w)\in E$, and $(v,w)\in E$ are not
independent. In Fig.~\ref{fig:threecases}, we sketch different
geometric situations between two nodes; positioning a node $w$ within
the transmission range of $u$ nicely illustrates the intuition behind
this fact.

\begin{figure}[htbp]
\begin{center}\includegraphics[scale=0.3,width=6cm,height=4cm]{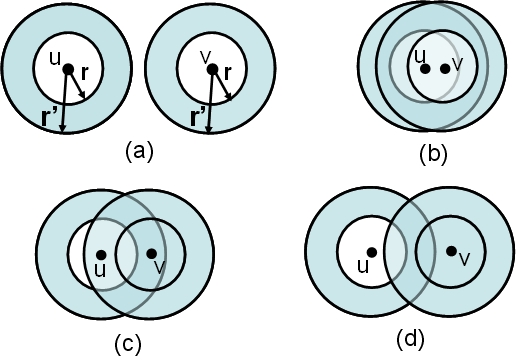}
\caption{Consider two relay nodes $u$ and $v$ of $G|_R$. Subfigure (a)
illustrates the situation when the two nodes are not connected and far
apart.  The other subfigures illustrate the following situations: (b)
$d(u,v)\leq r$, (c) $r<d(u,v)\leq r'$, and (d) $2r<d(u,v)\leq r+r'$.}
    \label{fig:threecases}
\end{center}
\end{figure}

However, certain edges in a connectivity graph are independent.
Indeed, all edges that are incident with a fixed (common) vertex are
independent, since the coordinates of the vertices in the underlying
quasi geometric random graph are chosen independently and uniformly at
random. Consequently, the random variables in the set $\{ C_{ij}\,|\,
j\in I\}$, where $i$ is fixed, are independent. We will take advantage
of this fact in our proof of the concentration result. To that end,
recall Chernoff's bound for sums of independent Bernoulli random
variables.

\begin{lemma}[Chernoff bound] \label{l:chernoff}
Let $X_1,\dots,X_m$ be independent Bernoulli random variables such
that $\Pr[X_k=1]=p'$ and $\Pr[X_k=0]=1-p'$, and let $X=\sum_{k=1}^m
X_k$.  For $0< \epsilon< 1$, we have
$$ \Pr[X\le (1-\epsilon)\E[X]]\le e^{-E[X]\epsilon^2/2}.$$
\end{lemma}
\smallskip
\begin{proof}
See, for instance,~\cite[p.~66]{Upfal05} for a proof of this well-known bound.
\end{proof}

\begin{lemma} \label{bigsumprob}
If $Pr[X+Y\leq a] \leq Pr[X\leq \frac{a}{2}] \bigcup Pr[Y\leq \frac{a}{2}]$.
\end{lemma}
\smallskip
\begin{proof}
This is because if $X>a/2$ and $Y>a/2$, definitely $X+Y>a$. So if $X+Y<a$,
at least one of $X$ and $Y$ must be less than $a/2$.
\end{proof}
This lemma is quite simple, but it turns out play a crucial role in the proof of following theorem.

\begin{theorem}\label{splitrandom}
For all cuts of size $k$, and all $\epsilon$, we have
$$ \Pr[C_k\le (1-\epsilon)\E[C_k]] \leq e^{-\big(\frac{\epsilon^2 (n-k) p' }{2}-\ln(k+1) \big)}$$
\end{theorem}

\begin{proof}
Let $C_k$ denote the capacity of an $s$-$t$-cut $(\{s\}\cup V_k;
\overline{V}_k\cup \{t\})$ in the connectivity graph.  Here $s$ is the
source node, $t$ is a terminal node, and $V_k\cup \overline{V}_k$ is a
partitition of the relay nodes into two disjoint sets $V_k$ and
$\overline{V}_k$ that respectively have cardinality $k=|V_k|$ and
$n-k=|\overline{V}_k|$.

We can reformulate equation~(\ref{eq:cut}) in the form
\begin{equation}\label{eq:altcut}
C_k = \sum_{i\in \overline{V}_k} C_{si} +
\sum_{j\in\vphantom{\overline{V}} V_k} \sum_{i\in
\overline{V}_k\cup\{t\}} C_{ji}.
\end{equation}
So following lemma ~\ref{bigsumprob} and the above formula,
if the event
\begin{equation}
C_k\le (1-\epsilon)\E[C_k]
\end{equation}
happens, then at least one of the following $k+1$ simpler events
must happen also
\begin{compactenum}[(i)]
\item $\sum_{i\in \overline{V}_k} C_{si}\le (1-\epsilon)E[C_k]/(k+1)$,
\item $\sum_{i\in \overline{V}_k\cup\{t\} } C_{ji}\le
(1-\epsilon)E[C_k]/(k+1) $,
\end{compactenum}

where $j\in V_k$. Since the left hand side of (i) and (ii) are sums of
independent Bernoulli random variables, we can use
Lemma~\ref{l:chernoff} to bound the probability of these events.
Therefore, we obtain the estimate
$$
\begin{array}{l}
\Pr[C_k\le (1-\epsilon)\E[C_k]]\\[1ex] \displaystyle \le
\Pr[\sum_{i\in \overline{V}_k} C_{si}\le (1-\epsilon) E[C_k]/(k+1)]\\
\displaystyle\phantom{\le}+ \sum_{j\in V_k} \Pr[ \sum_{i\in
\overline{V}_k\cup\{t\} } C_{ji}\le (1-\epsilon) E[C_k] /(k+1)]\\
[2ex] \displaystyle \le \Pr[\sum_{i\in \overline{V}_k} C_{si}\le (1-\epsilon) (k+1)(n-k) p' /(k+1)]\\
\displaystyle\phantom{\le}+ \sum_{j\in V_k} \Pr[ \sum_{i\in
\overline{V}_k\cup\{t\} } (C_{ji}\le (1-\epsilon) (k+1)\\ \hspace{4cm}(n-k+1) p' /(k+1))]\\
\\[2ex] \le \exp(-(n-k)p'\epsilon^2/2)\\
\phantom{\le}+k\exp(-(n-k+1)p'\epsilon^2/2)\\[1ex] \le
\exp(-((n-k)p'\epsilon^2/2 - \ln (k+1))).
\end{array}
$$
\end{proof}

In the next two theorems, we are going to show that the capacity of a
minimum cut is, with high probability, concentrated about the value
$np'=\E[C_0]$. Intuitively, it is not surprising that the bottleneck
is likely going to be the connection from the source to the relay
nodes, so the dissemination of the information is likely to be
limited.

\begin{theorem}
Let $G$ be a connectivity graph with one source node $s$, $n$ relay
nodes, and a set $T$ of terminal nodes. Then,
with probability $1-O(\tau/n^2)$, where $\tau=|T|$,
the network coding
capacity $C_{s,T}$ of $G$ is bounded from below by
$$ C_{s,T} \ge (1-\epsilon)\E[C_0],\quad\text{where}\quad
\epsilon=\sqrt{\frac{4\ln n}{p'(n-k)}},$$
where $p'$ satisfies (\ref{eq:conn_prob}).
\end{theorem}
\begin{proof}
Let $C_{min}(s,t)$ denote the capacity of a minimum $s$-$t$-cut.  Let us
assume further that this minimum cut has size $k$, that is,
$C_{min}(s,t)=C_k$.  Since $\E[C_k]\ge \E[C_0]$ holds for all $k$ in the
range $0\le k\le n$, we have
$$
\begin{array}{l@{}l}
\Pr[C_{min}&(s,t) <(1-\epsilon)\E[C_0]]\le
\Pr[C_k <(1-\epsilon)\E[C_k]]\\[1ex]
&\le \Pr[|C_k-\E[C_k]|>\epsilon\E[C_k]] \\[1ex]
&< 2\exp\left( (-(n-k)p'\epsilon^2/2 + \ln(k+1))\right),
\end{array}
$$ where the last inequality is due to
Theorem~\ref{splitrandom}.
Substituting the value of $\epsilon$ from the hypothesis yields
$$
\begin{array}{l@{\,\,}c@{\,\,}l}\Pr[C_{min} <(1-\epsilon)\E[C_0]] &<&
2\exp(-2\ln n)\\
&=& O(1/n^2).
\end{array}
$$ Consequently, the probability that the network coding capacity
$C_{s,T}$ will be below the value $(1-\epsilon)\E[C_0]$ can be bounded by
$$
\begin{array}{l@{}l}
\Pr[C_{s,T}&<(1-\epsilon)\E[C_0]] \\
& \le  \displaystyle
\Pr\!\left[\bigcup_{t\in T}(C_{min}(s,t)<(1-\epsilon)\E[C_0])\right] \\[1ex]
&\le \displaystyle
\sum_{t\in T}
\Pr\!\left[C_{min}(s,t)<(1-\epsilon)\E[C_0]\right] \\[1ex]
&= O(\tau /n^2),
\end{array}
$$
as claimed.
\end{proof}

We complement the above lower bound by a high-probability upper bound
on the network coding capacity.

\begin{theorem}
Let $G$ be a connectivity graph with one source node $s$, $n$ relay
nodes, and a set $T$ of terminal nodes. Then,
with probability $1-O(1/n^{4/3})$,
the network coding
capacity $C_{s,T}$ of $G$ is bounded from above by
$$ C_{s,T} \le (1+\epsilon)\E[C_0],\quad\text{where}\quad
\epsilon=\sqrt{\frac{4\ln n}{ \E[C_0]}},$$
where $p'$ satisfies (\ref{eq:conn_prob}).
\end{theorem}
\begin{proof}
If the network coding capacity $C_{s,T}$ exceeds the value
$(1+\epsilon)\E[C_0]$, then the capacity of any $s$-$t$-cut, for any $t\in T$,
must exceed that value as well; in particular,
the cut $(\{s\}; R\cup T)$ must have capacity
exceeding $(1+\epsilon)\E[C_0]$. Since we assume that the source node
is not directly connected to any terminal node, we obtain
$$
\begin{array}{l@{}l}
\Pr[&C_{s,T} > (1+\epsilon)\E[C_0]]\\
& \le \Pr[\sum_{r\in R} C_{sr} > (1+\epsilon)\E[C_0]] \\
& \le \Pr[|\sum_{r\in R} C_{sr} - \E[C_0]|> \epsilon\E[C_0]]
\end{array}
$$ The indicator random variables $C_{sr}$, with $r\in R$, are
mutually independent, as the location of the relay nodes are
independently and identically distributed in the unit square. Recall
that the Chernoff bound for independent identically distributed
indicator random variables $X_i$ with $\Pr[X_i=1]=p'$ is given by
$\Pr[|\sum_{i=1}^n X_i-np'|>t]< 2\exp(-t^2/3np')$.  Applying this
Chernoff bound to the indicator random variables $C_{sr}$ yields
$$
\begin{array}{l@{}l}
\Pr[|\sum_{r\in R}&  C_{sr} - \E[C_0]|> \epsilon\E[C_0]] \\
& < 2 \exp(-\epsilon^2 \E[C_0]^2/ (3\E[C_0]))\\
& = \displaystyle 2\exp\left(-\frac{4\ln n}{E[C_0]}\frac{\E[C_0]^2}{3\E[C_0]}\right)\\
& = O(n^{-{4/3}}),
\end{array}
$$
which proves the claim.
\end{proof}

\begin{remark}
Our results easily generalize to more general substrates of unit area
(not just unit squares), as long as the assumption holds
that the nodes are uniformly distributed over the area. The
concentration results are not affected by such a change, but the
connectivity probability $p'$ might be dramatically different.
For instance, if the area is a rectangle that is $\varepsilon$ high and
$1/\varepsilon$ wide, then $p'$ approaches $0$ as $\varepsilon$ approaches $0$.
\end{remark}

\section{Simulations and Experiments}\label{simulation}
We conducted simulations for various instances of $\mathcal{G}^{QRGG}$
using different parameters. Our simulation results support the high
probability bounds on the network coding capacity given in
Theorems~9 and~10.

In a first experiment, we determined the minimum capacity of an
$s$-$t$ cut for different instances of a connectivity graph in
$\mathcal{G}^{QRGG}$ with a fixed number of nodes. Fig.~\ref{fig:fig2}
shows the results of such an experiment with $n=200$ relay nodes. The
radio transmission range is chosen such that within a radius of
$r=0.1$ the connectivity is guaranteed and up to a radius of $r'=0.2$
one might get connected. The plot shows that the capacity of the
network is concentrated around the expected value of $13$ which is in
agreement with Theorem~9 and~10 for the above values of $n$, $r$
and $r'$.

\begin{figure}[htbp]
\begin{center}
\includegraphics[scale=0.4]{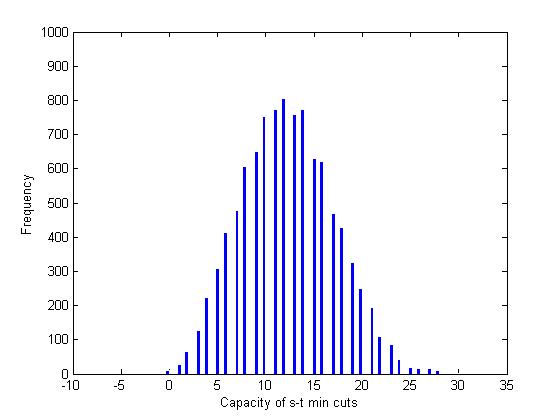}
    \caption{n=200, r=0.1, $r'=0.2$}
    \label{fig:fig2}
\end{center}
\end{figure}

Fig.~\ref{fig:fig3} shows the result of a second experiment.  This
time, the number of relay nodes is once again $n=200$, but the
transmission range is higher, namely the inner radius equals $r=0.13$
and outer radius equals $r'=0.18$. We generated random instances of
$\mathcal{G}^{QRGG}$ with these parameters and determined the minimum
cut. One can easily see that the capacity of the network is likely to
be higher, as expected.
\begin{figure}[htbp]
\begin{center}
\includegraphics[scale=0.35]{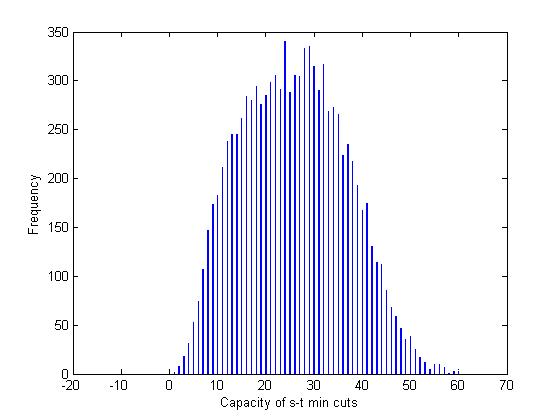}
    \caption{n=200, r=0.13, $r'=0.18$}
    \label{fig:fig3}
\end{center}
\end{figure}

For larger $n$, we could observe that the histograms become more
concentrated around the expected capacity of a minimum cut, as
predicted by our theory.

In a third series of experiments, we simulated the increase of
capacity of the minimum cut for different values of $r$ and $n$. In
this case, we also modeled the connectivity probability as a
decreasing function of distance, following Remark~\ref{rem:var_prob},
$$p=\left(1-\sqrt{\frac{d(i,j)^{2}-r^{2}}{r'^{2}-r^{2}}}\right)
p_{connection},$$ where $d(i,j)$ is the Euclidean distance between any
two nodes $i$ and $j$ such that $r<|d(i,j)|<r'$, and $p_{connection}$
is a probability that accounts for the interference
noise in communication.
\begin{figure}[htbp]
\begin{center}
\includegraphics[scale=0.35]{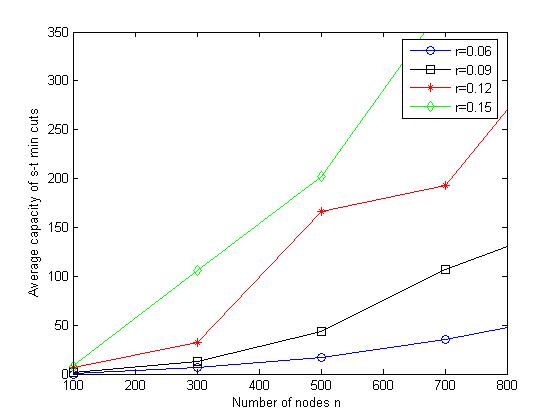}
    \caption{The capacity of s-t minimum cuts with different values of
    $n$ and $r$}
    \label{Fig:cut_n_r}
\end{center}
\end{figure}

As it can be seen from Fig. \ref{Fig:cut_n_r}, the value of the
capacity grows more rapidly for lower values of $r$. This is intuitive
because in that case not many nodes are connected for small values of
$n$. As we increase $n$ but keep $r$ constant, the capacity of the
minimum cut must increase, since more and more nodes are packed in the
same area.

\section{Conclusion}
\label{conclusion} We modeled a quasi wireless random network and showed that the
capacity of the minimum cut of network coding is concentrated around the value
$np'=\E[C_0]$. Unlike prior works, we obtained high probability bounds for this
model. More realistic models (for example, when the probability of connectivity
drops exponentially with distance to account
for signal attenuation) can be easily incorporated into our framework
without changing the theory in a significant way.

\IEEEtriggeratref{3}
\bibliographystyle{plain}

\end{document}